\documentclass[USenglish, a4paper, numberwithinsect, cleveref]{lipics-v2021}
\usepackage{xspace}

\newcommand\G{\ensuremath{\mathcal{G}}\xspace}
\newcommand\TC{\textsf{TC}\xspace}

\usepackage{algorithm}
\usepackage{cleveref}
\usepackage{tikz}
\usepackage[noend]{algorithmic}

\usetikzlibrary{snakes,arrows.meta,calc,shapes.symbols,shapes.arrows,decorations.pathreplacing,calc,tikzmark,fit}
\nolinenumbers

\usepackage{xcolor}

\xdefinecolor{green}{rgb}{0,0.7,0}
\tikzset{defnode/.style={draw, circle, gray, inner sep=0, minimum size=2.3}}

\usepackage{environ}
\NewEnviron{killcontents}{}

\title{Structural Lemmas on Temporal Connectivity}
\subtitle{Revisiting gossip literature and beyond}

\author{Daniele Carnevale}{Department of Computer Science, University of Geneva, Switzerland}{daniele.carnevale@unige.ch}{https://orcid.org/0009-0006-3997-869X}{Supported by the Swiss NSF project RECAPT (200021-236640).}
\author{Arnaud Casteigts}{Department of Computer Science, University of Geneva, Switzerland}{arnaud.casteigts@unige.ch}{https://orcid.org/0000-0002-7819-7013}{Supported by the Swiss NSF project RECAPT (200021-236640).}
\author{David Schindl}{Haute École de Gestion, Geneva, Switzerland}{david.schindl@hesge.ch}{https://orcid.org/0000-0002-7009-5530}{}

\authorrunning{Daniele Carnevale, Arnaud Casteigts, David Schindl}
\Copyright{Daniele Carnevale, Arnaud Casteigts, David Schindl}

\ccsdesc[500]{Mathematics of computing~Discrete mathematics}
\ccsdesc[500]{Theory of computation~Graph algorithms analysis}

\keywords{Temporal connectivity, temporal cycles, disjoint paths, minimal TC graphs, happy temporal graphs.}

\begin{document}
\maketitle
\begin{abstract}
  This paper presents several lemmas on the structure of temporal connectivity in temporal graphs. Some of these lemmas are adapted from the literature on gossip from the 70-90's and reformulated in the context of temporal graph theory. Many others are original. For ease of presentation and for avoiding lengthy case distinctions, we formulate all the lemmas in the setting of simple and proper (a.k.a, happy) temporal graphs, discussing occasionally their generality beyond this setting.
\end{abstract}

\section{Introduction}

A temporal graph is a labeled graph $\G=(V,E,\lambda)$ where $G=(V,E)$
is a standard finite graph called the \emph{footprint}. In general,
the labeling function $\lambda : E \to 2^{\mathbb{N}}$ assigns one or
several \emph{time labels} to each edge of $E$, interpreted as
presence times (we will further restrict $\lambda$). The largest time label $L$ is the \emph{lifetime} of
$\G$. Reachability in temporal graphs is defined in terms of walks that
traverse the edges chronologically. Namely, a \emph{temporal walk} is
a sequence $\langle(e_i, t_i)\rangle$ such that $\langle e_i \rangle$
is a walk in $G$, $\langle t_i \rangle$ is nondecreasing, and
$t_i \in \lambda(e_i)$ for all $i$ in the sequence. If $\langle t_i \rangle$ is increasing, then the temporal walk is \emph{strict}. Temporal paths are defined analogously, forbidding the repetition of intermediate vertices.

The distinction between strict and non-strict walks gives rise to distinct versions of many problems and properties. A common restriction is to assume that the labeling $\lambda$ is \emph{proper}; namely, two edges incident to the same vertex cannot share a label. In proper temporal graphs, strict and non-strict temporal walks coincide.
Another common restriction is to consider that the edges have a single time label, referred to as \emph{simple} temporal graphs. Although restricted, temporal graphs that are both proper and simple (called \emph{happy}) already capture many interesting phenomena in temporal graphs. Furthermore, negative results obtained in happy graphs apply in general, making them a natural target for hardness reduction. For positive results, considering happy graphs comes with an inevitable loss of generality. However, positive result on temporal graph should hold at least in this very specific case. Happy graphs are thus a relevant prototype for studying the general theory.

In this work, we identify a set of structural lemmas on happy temporal graph \G, with a focus on temporal connectivity, where all the vertices can reach each other through temporal walks (class \TC). In particular, we focus on three key properties which are \emph{minimality} (none of the edges can be removed without impacting reachability), \emph{acyclicity} (absence of temporal cycles), and \emph{two-path}-freeness (at most one temporal walk between every ordered pair of vertices). It turns out that these three properties have interesting connections with each other. They also impose strong restrictions on the \emph{structure} of the temporal graph, which is a rare feature in temporal graphs.

Another reason why these properties are relevant is that they were already studied in the literature of gossip in the 70-90's (in particular,~\cite{tijdeman71,hajnal,baker72,bumby,west_acyclic1,west_disjoint,seress,labahn93,labahn95,krumme}), many results of which have gone unnoticed so far by the temporal graph community. Some of our lemmas are reformulation of existing results in this literature. Sometimes, the properties do not appear explicitly in these papers, we extract them from arguments found in the proof of other results. On the way, we prove a number of original lemmas and we introduce some tools of independent interest, such as \emph{compressed temporal graphs} and \emph{contiguous caterpillars}.

Regarding the lemmas themselves, we opt for breaking them into a myriad of small lemmas rather than final theorems whose proof would combine all arguments. We do so to make them as modular as possible, and also to pave the way towards a potential subsequent formalization using a proof assistant like Coq, Lean, or Isabelle. This also results in lemmas whose individual proofs are often easier to follow.

\subsection{Related works}

This preliminary version of the paper was completed hastly before a Dagstuhl seminar, so that one of the authors can present the results. The content of this section will be added in a subsequent version of the paper. In particular, we will review here a number of results that have been obtained in the temporal graph community, focusing on \emph{structural} properties. We will also give a summary and time line of the main results in gossip theory. 

\section{Preamble}

From this point on, unless otherwise mentioned, all the temporal graphs we consider are happy and we do not recall this fact. As most of the properties we consider are temporal, we often omit the adjective ``temporal''. For example, a temporal graph is just a graph, and its footprint is a \emph{static} graph. Furthermore, as happy graphs are simple, we abuse the terminology by referring to both an edge and its label just as an edge.

\subsection{Further definitions}

Let $\G=(V,E,\lambda)$. For a vertex $v$, we denote by $e^-(v)$ the minimum edge of $v$ and by $n^-(v)$ the corresponding neighbor (using the adjectives minimum or earliest interchangeably). We define $e^+(v)$ and $n^+(v)$ analogously.
We write $E^- = \{e^-(v) \mid v \in V\}$ and  $E^+ = \{e^+(v) \mid v \in V\}$ for the set of edges that are minimum (maximum) for at least one vertex. 

If there is a temporal walk from vertex $u$ to vertex $v$, we say that $u$ can reach $v$, and write $u \leadsto v$. The reachability graph $\mathcal{R}(\G)$ is a static directed graph $(V,A)$ whose set of arcs $A=\{(u,v) : u\leadsto v\}$ represent the existence of temporal walks in $\G$. A temporal graph $\G$ is \emph{minimal} if for all edge $e$ in $E$, $\mathcal{R}(\G \setminus \{e\}) \ne \mathcal{R}(\G)$, i.e. none of the edges can be removed without affecting reachability.

A \emph{temporal cycle} is a (nonempty) temporal walk from a vertex to itself. A small subtlety should be mentioned here. As already said, strict and non-strict walks coincide in happy graphs (and more generally proper graphs). The only exception is that a non-strict walk may travel back and forth along the same edge consecutively, creating an artificial cycle of length $2$. To avoid this degenerate case, we require temporal walks to be simple (i.e., do not repeat edges) throughout the paper.
Then a graph $\G$ is \emph{acyclic} if it has no temporal cycles.
A graph is two-path-free if for all $u,v$, there is at most one temporal walk from~$u$ to~$v$. Whether~$u$ and~$v$ are required to be distinct in this definition does not affect the lemmas, thus we leave it unspecified. 
We will often consider minimal, acyclic, or two-path-free graphs that are also temporally connected (in \TC), discussing the extra properties that these graphs must satisfy.

The reachability notation $\leadsto$ is extended in natural ways to work between vertices and edges, between edges, and between sets of both. In particular, $v\leadsto e$ means that there exists a temporal walk that starts at $v$ and traverses the edge $e$ at some point. Similarly, $e\leadsto v$ means that there is a walk that traverses $e$ at some point and ends at $v$. We write $e \leadsto e'$ if there exists a temporal path that traverses $e$ and $e'$ in this order (not necessarily consecutively). If $e\not\leadsto e'$ and $e'\not\leadsto e$, then $e$ and $e'$ are called \emph{incomparable}.

Let $V'\subseteq V$ and $E'\subseteq E$, then $e \leadsto V'$ means $e\leadsto v$ for all $v$ in $V'$, and $v\leadsto E'$ means that $v\leadsto e$ for all $e$ in $E'$. 
If $v\leadsto V$ (resp. $e\leadsto V$), then $v$ is a \emph{source} vertex (resp. $e$ is a \emph{source} edge).
If $V\leadsto v$ (resp. $V \leadsto e$), then $v$ is a \emph{sink} vertex (resp. $e$ is a \emph{sink} edge).
A \emph{pivot edge} is an edge that is both a source edge and a sink edge. A graph that has a pivot edge is called \emph{pivotable}.

\subsection{Time compression}

If one does not care about the exact timing, then the reachability of a temporal graph depends only on the local \emph{ordering} of the edges at the vertices. For example, the graphs $\G_1$ and $\G_2$ in Figure~\ref{fig:equivalent} have the same reachability in a very strong sense: there is a bijection between the temporal walks in $\G_1$ and the ones in $\G_2$. 
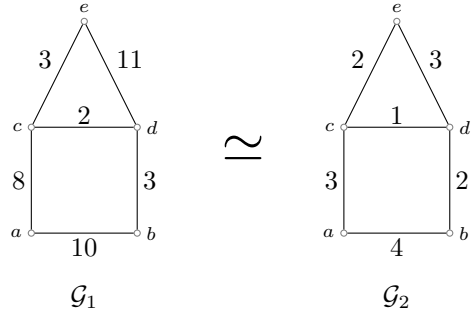
\begin{figure}[ht]
\centering
  \begin{tikzpicture}[scale=1.4]
    \tikzstyle{every node}=[defnode]
    \path (0,0) node (a) {};
    \path (1,0) node (b) {};
    \path (0,1) node (c) {};
    \path (1,1) node (d) {};
    \path (.5,2) node (e) {};
    \tikzstyle{every node}=[font=\scriptsize]
    \path (a) node[left] {$a$};
    \path (b) node[right] {$b$};
    \path (c) node[left] {$c$};
    \path (d) node[right] {$d$};
    \path (e) node[above] {$e$};
    \tikzstyle{every node}=[inner sep=2pt]
    \draw (a) -- node[below] {10}(b);
    \draw (a) -- node[left] {8}(c);
    \draw (b) -- node[right] {3}(d);
    \draw (c) -- node[above=-1pt] {2}(d);
    \draw (c) -- node[above left] {3}(e);
    \draw (d) -- node[above right] {11}(e);
    \path (.5,-.6) node (g1) {$\G_1$};
    \path (2,.8) node(simeq){\huge $\simeq$};
  \end{tikzpicture}
  ~~~
  \begin{tikzpicture}[scale=1.4]
    \tikzstyle{every node}=[defnode]
    \path (0,0) node (a) {};
    \path (1,0) node (b) {};
    \path (0,1) node (c) {};
    \path (1,1) node (d) {};
    \path (.5,2) node (e) {};
    \tikzstyle{every node}=[font=\scriptsize]
    \path (a) node[left] {$a$};
    \path (b) node[right] {$b$};
    \path (c) node[left] {$c$};
    \path (d) node[right] {$d$};
    \path (e) node[above] {$e$};
    \tikzstyle{every node}=[inner sep=2pt]
    \draw (a) -- node[below] {4}(b);
    \draw (a) -- node[left] {3}(c);
    \draw (b) -- node[right] {2}(d);
    \draw (c) -- node[above=-1pt] {1}(d);
    \draw (c) -- node[above left] {2}(e);
    \draw (d) -- node[above right] {3}(e);
    \path (.5,-.6) node (g1) {$\G_{2}$};
  \end{tikzpicture}
  \caption{Two equivalent graphs, one of which is compressed (right).}
  \label{fig:equivalent}
\end{figure}
Among all the graphs equivalent to these two, $\G_2$ is in a sense canonical, as its labels are as small as they can be. This is captured more formally by the following definition from~\cite{CCC25}.

\begin{definition}[Contiguity property]
  \label{def:contiguity}
  A graph is (time-)compressed if for every edge~$e$ with $\lambda(e) > 1$, there exists an edge $e'$ adjacent to $e$ such that $\lambda(e')=\lambda(e)-1$.
\end{definition}

As far as (time agnostic) reachability is concerned, compressed graphs capture all the possible cases. We can thus restrict our attention to these graphs without loss of generality, while enjoying the contiguity property that simplify some proofs. In case one needs to compress a graph explicitly, we provide a polynomial-time algorithm that transforms a graph into its (unique) compressed version in~\Cref{sec:compression}. This algorithm is not required in the present paper. 
The following definitions apply in the context of compressed graphs (they would not make much sense otherwise).

\begin{definition}[Unique label]
  \label{def:unique}
  The label of an edge $e$ is unique if for all $e'\ne e$, $\lambda(e')\ne \lambda(e)$. 
\end{definition}

\begin{definition}[Full-range graphs]
  \label{def:full-range}
  A graph is full-range if all the labels are unique. 
\end{definition}

\subsection{Contiguous caterpillars}

If the graph is compressed, one can unfold the contiguity property recursively from any edge, obtaining a useful structure we refer to as a contiguous caterpillar. More precisely, let $uv\in E$ be an edge of $\G$ and initialize $\G'$ as a temporal graph containing $u$, $v$, and the edge $uv$ (with same label). If $\lambda(uv)=1$, we are done. Otherwise, as $\G$ is compressed, there must exist another edge incident to either $u$ or $v$, say $uw$, such that $\lambda(uw)=\lambda(uv)-1$. Add $uw$ and $w$ to $\G'$, then repeat the procedure with respect to $uw$. If at any point there are two candidate edges, choose one arbitrarily. If at any point the new node $w$ is already in $\G'$, create a new copy of $w$ instead and add it to $\G'$. In the end of the process, we obtain a temporal graph $\G'$ whose footprint is a caterpillar.

  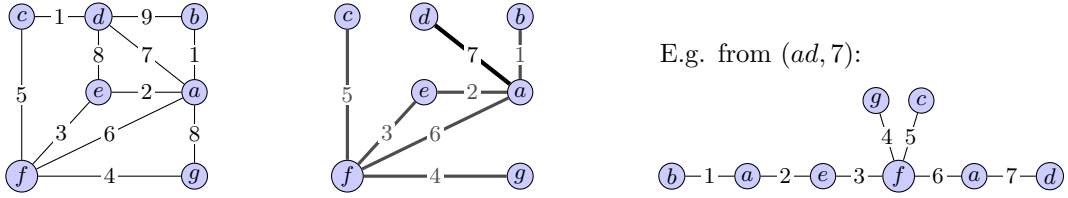
\begin{figure}[h]
  \begin{tikzpicture}[scale=.5]
    \tikzstyle{every node}=[circle,inner sep=1pt,minimum size=9.5pt,draw,fill=blue!20,font=\footnotesize]
    \node (0) at (1.87, -.70) {$a$};
    \node (1) at (1.87, 1.3) {$b$};
    \node (2) at (-2.70, 1.30) {$c$};
    \node (3) at (-0.67, 1.30) {$d$};
    \node (4) at (-0.67, -.70) {$e$};
    \node (5) at (-2.70, -2.92) {$f$};
    \node (6) at (1.87, -2.92) {$g$};
    \tikzstyle{every node}=[fill=white,inner sep=1pt,pos=.5,font=\footnotesize]
    \draw (0) -- node {1} (1);
    \draw (0) -- node {2} (4);
    \draw (0) -- node {6} (5);
    \draw (0) -- node {7} (3);
    \draw (0) -- node {8} (6);
    \draw (1) -- node {9} (3);
    \draw (2) -- node {1} (3);
    \draw (2) -- node {5} (5);
    \draw (3) -- node {8} (4);
    \draw (4) -- node {3} (5);
    \draw (5) -- node {4} (6);
    \end{tikzpicture}
    \hfill
  \begin{tikzpicture}[scale=.5]
    \tikzstyle{every node}=[circle,inner sep=1pt,minimum size=9.5pt,draw,fill=blue!20,font=\footnotesize]
    \node (0) at (1.87, -.70) {$a$};
    \node (1) at (1.87, 1.3) {$b$};
    \node (2) at (-2.70, 1.30) {$c$};
    \node (3) at (-0.67, 1.30) {$d$};
    \node (4) at (-0.67, -.70) {$e$};
    \node (5) at (-2.70, -2.92) {$f$};
    \node (6) at (1.87, -2.92) {$g$};
    \tikzstyle{every node}=[fill=white,inner sep=1pt,pos=.5,font=\footnotesize]

      \tikzstyle{every path}=[very thick,black!70]
      \draw (0) -- node {1} (1);
      \draw (0) -- node {2} (4);
      \draw (0) -- node {6} (5);
      \draw[ultra thick,black] (0) -- node {7} (3);
      \draw (2) -- node {5} (5);
      \draw (4) -- node {3} (5);
      \draw (5) -- node {4} (6);
  \end{tikzpicture}
  \hfill
  \begin{tikzpicture}
    \tikzstyle{every node}=[circle,inner sep=1pt, minimum size=9.5pt,draw, fill=blue!20,font=\footnotesize]
    \node (b) at (0,0) {$b$};
    \node (a) at (1,0) {$a$};
    \node (e) at (2,0) {$e$};
    \node (f) at (3,0) {$f$};
    \node (g) at (2.7,1) {$g$};
    \node (c) at (3.3,1) {$c$};
    \node (a2) at (4,0) {$a$};
    \node (d) at (5,0) {$d$};

    \tikzstyle{every node}=[fill=white,inner sep=1pt,font=\footnotesize]
    \path (-.2,1.6) node[right] {\normalsize E.g. from $(ad,7)$:};
    \draw (a) -- node {1} (b);
    \draw (a) -- node {2} (e);
    \draw (a2) -- node {6} (f);
    \draw (a2) -- node {7} (d);
    \draw (c) -- node {5} (f);
    \draw (e) -- node {3} (f);
    \draw (f) -- node {4} (g);
  \end{tikzpicture}
  \caption{\label{fig:caterpillar} Example of a contiguous caterpillar.}
\end{figure}

The contiguous caterpillar of an edge is not necessarily unique (though in this example, it is). In general, an edge with label $t$ could have two adjacent edges with label $t-1$ (one at each endpoint). Regardless, all caterpillars enjoy the following properties:

\begin{lemma}[Basic properties of contiguous caterpillars]\label{lem:caterpillar}~\\
  (i) For an edge $e$, a caterpillar from $e$ contains exactly one edge for each label in $[1,\lambda(e)]$.\\
  (ii) For all edges $e_1,e_2$ in the caterpillar such that $\lambda(e_1) < \lambda(e_2)$, $e_1 \leadsto e_2$.\\  
  (iii) The lifetime $L$ is the number of edges in the caterpillar of the largest edge $e=\arg\max(\lambda)$.
\end{lemma}

In the particular case of a full-range graph, the caterpillar of the largest edge contains all the edges of the graph.
\begin{figure}[h]
  \centering
    \begin{tikzpicture}[yscale=.6]
      \node[defnode] (0) at (0, 1) {};
      \node[defnode] (1) at (1, 0) {};
      \node[defnode] (2) at (1, 2) {};
      \node[defnode] (3) at (2, 1) {};
      \node[defnode] (4) at (3, 1) {};
      \tikzstyle{every node}=[fill=white,inner sep=1pt,pos=.5,font=\footnotesize]
      \draw (0) -- node {1} (1);
      \draw (0) -- node {2} (2);
      \draw (1) -- node {3} (2);
      \draw (1) -- node {6} (3);
      \draw (2) -- node {4} (3);
      \draw (3) -- node {5} (4);
    \end{tikzpicture}
    \quad
    \begin{tikzpicture}
      \path (0,-.4) coordinate (bidon);
      \node[defnode] (0) at (2, 0) {};
      \node[defnode] (1) at (1, 0) {};
      \node[defnode] (11) at (3, 1) {};
      \node[defnode] (111) at (5, 0) {};
      \node[defnode] (2) at (3, 0) {};
      \node[defnode] (3) at (4, 0) {};
      \node[defnode] (4) at (4, 1) {};
      \tikzstyle{every node}=[fill=white,inner sep=1pt,pos=.5,font=\footnotesize]
      \draw (0) -- node {1} (1);
      \draw (0) -- node {2} (2);
      \draw (11) -- node {3} (2);
      \draw (111) -- node {6} (3);
      \draw (2) -- node {4} (3);
      \draw (3) -- node {5} (4);
    \end{tikzpicture}
    \caption{\label{fig:full-range}Example of a full-range graph and a caterpillar from the largest edge.}
  \end{figure}
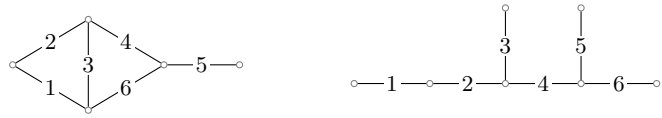

Finally, note that another notion of caterpillar was used in the gossip literature~\cite{west_acyclic1}. The version from~\cite{west_acyclic1} is different, it is a standard caterpillar which does not repeat vertices.

\section{The lemmas}

From this point on, the considered temporal graphs are always \emph{happy} and \emph{compressed}, and we do not recall this fact. We also refer to contiguous caterpillars just as caterpillars. We start by a few general lemmas. Then, we successively examine properties of acyclic graphs, acyclic \TC graphs, two-path-free graphs, two-path-free \TC graphs, \TC graphs, minimal graphs, and minimal \TC graphs. We also discuss the relation of some of these with pivotability.

\subsection{General lemmas}

Reachability among vertices is not transitive in temporal graphs. However, reachability among edges is transitive. This and others properties are given in Lemma~\ref{lem:basic-reachability}.
\begin{lemma}[Basic reachability facts]
  \label{lem:basic-reachability}~\\
  (i) If $e_1 \leadsto e_2$ and $e_2 \leadsto e_3$, then $e_1 \leadsto e_3$.\\
  (ii) If $v \leadsto e_1$ and $e_1 \leadsto e_2$, then $v \leadsto e_2$.\\
  (iii) If $v \leadsto e$ and $e \leadsto v'$, then $v \leadsto v'$.\\
  (iv) If $v\leadsto u$, then $e^-(v) \leadsto u$; $v \leadsto e^+(u)$; and $e^-(v) \leadsto e^+(v)$.\\
  (v) If $\G \in \TC$, then for all $u,v \in V$, $e^-(u)\leadsto v;u \leadsto e^+(v)$; and $e^-(u) \leadsto e^+(v)$.\\
  (vi) If $e_1$ and $e_2$ are distinct edges, then either $e_1 \not\leadsto e_2$ or $e_2 \not\leadsto e_1$.
\end{lemma}

\begin{proof}
  Denote by $W_{x,y}$ a temporal walk from $x\in V\cup E$ to $y\in V\cup E$ witnessing $x\leadsto y$, and let $e_i=(u_iv_i,t_i)$ be an edge.\\
  (i) Let $W_{e_1,e_2}$ and $W_{e_2,e_3}$ be two temporal walks. Either these walk traverse $e_2$ in the same direction, or they traverse it in opposite directions. In the first case, remove $e_2$ from $W_{e_1,e_2}$ and compose the resulting walk with $W_{e_2,e_3}$. In the second case, remove $e_2$ from both walks and compose the resulting walks. In both cases, we obtain a temporal walk from $e_1$ to $e_3$.\\
  (ii) Same proof as (i), substituting $e_1,e_2,e_3$ with $v,e_1,e_2$, respectively.\\
  (iii) Same proof as (i), substituting $e_1,e_2,e_3$ with $v,e,v'$, respectively.\\
  (iv) Unless the first edge $v\leadsto u$ is already $e^-(v)$, we obtain $W_{e^-(v),u}$ by traversing first the edge $e^-(v)$ from $n^-(v)$ to $v$ and then continuing on $v\leadsto u$. Similarly, we obtain $v \leadsto e^+(u)$ by appending $e^+(u)$ to $W_{v,u}$. Finally, we obtain $e^-(v) \leadsto e^+(v)$ by combining these extensions.\\
  (v) Direct consequence of (iv).\\
  (vi) If $e_1\leadsto e_2$ and $e_2 \leadsto e_1$, then $\lambda(e_1)<\lambda(e_2)$ and $\lambda(e_2) < \lambda(e_1)$ (contradiction).\qed
\end{proof}

\begin{lemma}{
    \label{lem:unique}
    Let $e,e'$ be two distinct edges. If $\lambda(e)$ is unique and $\lambda(e) < \lambda(e')$, then $e \leadsto e'$.}
\end{lemma}
\begin{proof}
  Any caterpillar from $e'$ must have an edge with label $\lambda(e)$. As this label is unique, this implies that $e$ belongs to the caterpillar of $e'$, thus by Lemma~\ref{lem:caterpillar}~(ii), $e\leadsto e'$.\qed
  \end{proof}

\subsection{Acyclic graphs}

\begin{lemma}
  Acyclic graphs have a triangle-free footprint.
\end{lemma}
\begin{proof}
  Any triangle in the footprint induces a temporal cycle for one of the three vertices.\qed
\end{proof}

\begin{lemma}
  \label{lem:lifetime-n}
  If $\G$ is acyclic, then $L(\G) < n$.
\end{lemma}
\begin{proof}
  By contradiction, suppose that $L(\G) \ge n$. Then, by Lemma~\ref{lem:caterpillar}~(iii), a caterpillar from the largest edge has $L(\G) \ge n$ edges, which implies at least $n+1$ vertices, thus at least two vertices of the caterpillar correspond to the same original vertex in $\G$. This vertex can reach itself through the caterpillar, thus it has a temporal cycle.\qed
\end{proof}

\subsection{Acyclic \TC graphs}

In addition to the properties of acyclic graphs, acyclic \TC graphs satisfy the following additional properties.

\begin{lemma}
  \label{lem:extremal-matching}
  If $\G$ is acyclic and \TC, then $E^-$ and $E^+$ are two perfect matchings. If $n>2$, these matchings are also disjoint.\\
  (see also Theorem 1 in~\cite{baker72} and Lemma 2.3 in~\cite{west_acyclic1}.)
\end{lemma}
\begin{proof}
  By contradiction, let $\G$ be an acyclic \TC graph such that two adjacent edges $uv$ and $vw$ are in $E^-$. Only one of $uvw$ and $wvu$ is a temporal walk, wlog say $uvw$ (thus $e^{-}(v)=uv$ and $e^{-}(w)=vw$, see the picture). As \G is \TC, $w$ must have a temporal walk to $u$. Either this walk starts with $vw$, or it starts with another edge of $w$. In the first case, one can compose $uv$ with the part of this walk from $v$ (because $uv=e^{-}(v)$). In the second case, one can compose $uvw$ with this walk (because $vw=e^{-}(w)$). In both cases, this creates a temporal cycle for $u$. (The proof is symmetric for $E^+$.) For the second part, suppose that $n>2$ and there is an edge $uv \in E^- \cap E^+$. There must exist a vertex $w\notin\{u,v\}$. As $e\in E^+$ and $\G\in \TC$, we have by~\Cref{lem:basic-reachability}~(v) that $w\leadsto e$, and as $e\in E^-$, we have $e\leadsto w$. Now, because $w\notin\{u,v\}$, $w \leadsto e \leadsto w$ implies a temporal cycle at $w$.

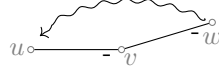
\begin{center}
    \begin{tikzpicture}[scale=1.2]
      \tikzstyle{every node}=[defnode]
      \path (0,0) node[label=left:$u$] (a){};
      \path (1,0) node[label=below right:$v$] (b){};
      \path (2,.3) node[label=below:$w$] (c){};
      \path (.1,.15) coordinate (aa);
      \tikzstyle{every node}=[]
      \draw (a) -- node[pos=.87,yshift=-3pt] {\normalsize -} (b);
      \draw (b) -- node[pos=.84,yshift=-3pt] {\normalsize -} (c);
      \draw (c) edge[decorate, decoration={snake, amplitude=.3mm, segment length=3mm},bend right=30,shorten <=2pt,->] (aa);  
    \end{tikzpicture}
  \end{center}\qed
\end{proof}

\begin{lemma}
  If $\G$ is acyclic and \TC, then $n$ is even.
\end{lemma}
\begin{proof}
  Follows directly from the fact that $E^-$ is a \emph{perfect} matching (\Cref{lem:extremal-matching}).\qed
\end{proof}

\begin{lemma}
  If $\G$ is acyclic and \TC, then its footprint is biconnected (has no cut vertex).
\end{lemma}
\begin{proof}
    Let $u$ be a vertex whose removal would disconnect the footprint into at least two nonempty sets of vertices $V_1$ and $V_2$. Wlog, suppose that $v=n^+(u)$ belongs to $V_2$. As $\G$ is \TC, $v$ must be able to reach $V_1$ through $u$, which implies an arrival at $u$ \emph{before} its last edge $uv$. Composing such walk with $uv$ give a temporal cycle for $v$ (see the picture).
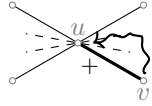
\begin{center}
  \begin{tikzpicture}
    \tikzstyle{every node}=[defnode]
    \path (1,0) node[label=above:$u$] (c){};
    \path (c)+(30:1) node (r1){};
    \path (c)+(10:.7) coordinate (r2){};
    \path (c)+(-10:.7) coordinate (r3){};
    \path (c)+(-30:1) node[label=below:$v$] (r4){};
    \path (c)+(150:1) node (l1){};
    \path (c)+(170:.7) coordinate (l2){};
    \path (c)+(190:.7) coordinate (l3){};
    \path (c)+(210:1) node (l4){};
    \path (c)+(.2,0) coordinate (cc){};

    \draw (c)--(r1);
    \draw[dashed] (c)--(r2);
    \draw[dashed] (c)--(r3);
    \draw (c)--(r4);
    \draw (c)--(l1);
    \draw[dashed] (c)--(l2);
    \draw[dashed] (c)--(l3);
    \draw (c)--(l4);

      \tikzstyle{every node}=[inner sep=1pt,font=\footnotesize]
      \draw[very thick] (c)-- node[below=1.5pt,pos=.15] {+}(r4);
        \draw (r4) edge[thick,decorate, decoration={snake, amplitude=.3mm, segment length=3mm},bend right=30,shorten <=2pt] (r2);
        \draw (r2) edge[thick,decorate, decoration={snake, amplitude=.3mm, segment length=3mm},bend right=5,->] (cc);
    \end{tikzpicture}
  \end{center}\qed
\end{proof}

\begin{lemma}
  If $\G$ is acyclic and \TC, then $L(\G) \le n/2$.
\end{lemma}
\begin{proof}
  Suppose $L(\G) \ge n/2+1$, then a caterpillar $\mathcal{C}$ from the largest edge has at least $n/2+2$ vertices. As \G is acyclic, these vertices are distinct in $\G$. We also have by~\Cref{lem:extremal-matching} that $E^-$ is a perfect matching. Now augment the caterpillar with all the edges of $E^-$ that have at least one endpoint in it (see the picture). Note that the first edge of $\mathcal{C}$ was also in $E^-$, we consider it as attached to the second node. As a perfect matching has $n/2$ edges, we now have at most $n/2$ edges from $E^-$ attached to at least $n/2+1$ vertices of $\mathcal{C}$, thus at least one edge of $E^-$ is attached twice. Call $e_1$ the first occurence and $e_2$ the second occurrence of this edge in the augmented caterpillar (see the picture on the right). Since the vertices of $\mathcal{C}$ are distinct, the other endpoint of $e_1$ (in white) must correspond to the endpoint of $e_2$ that is in $\mathcal{C}$. This vertex can reach itself through the caterpillar.
  \medskip
  
\hfill
    \begin{tikzpicture}
    \tikzstyle{every node}=[circle,inner sep=1pt, minimum size=4pt,draw, fill=blue!20,font=\footnotesize]
    \node[fill=white] (b) at (0,0) {};
    \node (a) at (1,0) {};
    \node (e) at (2,0) {};
    \node (f) at (3,0) {};
    \node (g) at (2.7,1) {};
    \node (c) at (3.3,1) {};
    \node (a2) at (4,0) {};
    \node (d) at (5,0) {};

    \tikzstyle{every node}=[fill=white,inner sep=1pt,font=\footnotesize]
    \draw (a) -- node {2} (e);
    \draw (a2) -- node {6} (f);
    \draw (a2) -- node {7} (d);
    \draw (c) -- node {5} (f);
    \draw (e) -- node {3} (f);
    \draw (f) -- node {4} (g);
      \tikzstyle{every node}=[draw,circle,inner sep=1pt]
      \path (b)+(0,-.2) coordinate (bb);
      \path (a)+(0,-.2) coordinate (aa);
      \path (a2)+(-.5,.5) node (a2a2){};
      \path (c)+(-.5,.5) node (cc){};
      \path (d)+(-.5,.5) node (dd){};
      \path (e)+(-.5,.5) node (ee){};
      \path (f)+(-.5,.5) node (ff){};
      \path (g)+(-.5,.5) node (gg){};
      \tikzstyle{every path}=[green,dash pattern=on 2pt off 1pt,thick]
      \draw (a) -- node[draw=none,fill=white,inner sep=.6pt,font=\footnotesize] {1} (b);
      \draw (cc) -- (c);
      \draw (dd) -- (d);
      \draw (ee) -- (e);
      \draw (ff) -- (f);
      \draw (a2a2) -- (a2);
      \draw (gg) -- (g);
      \draw (ee) -- (e);
  \end{tikzpicture}
  \hfill
    \begin{tikzpicture}
    \tikzstyle{every node}=[circle,inner sep=1pt, minimum size=4pt,draw, fill=blue!20,font=\footnotesize]
    \node (e) at (2,0) {};
    \node (f) at (3,0) {};
    \path (2,-.2) coordinate (bidon);

    \tikzstyle{every node}=[fill=white,inner sep=1pt,font=\footnotesize]
    \path (e) node[above right=6pt,yshift=4pt,rotate=-40]{\large $\leadsto$};
    \tikzstyle{every node}=[draw,circle,inner sep=1pt]
    \path (e)+(0,.7) node (ee){};
    \path (f)+(0,.7) node (ff){};
    \tikzstyle{every path}=[green,dash pattern=on 2pt off 1pt,thick]
    \draw (ee) -- node[draw=none,left=-1pt,yshift=2pt,font=\footnotesize]{$e_1$} (e);
    \draw (ff) -- node[draw=none,right,yshift=2pt,font=\footnotesize]{$e_2=e_1$} (f);
    \path (ee) node[draw=none,black,font=\scriptsize,above=1pt]{$x$};
    \path (f) node[draw=none,black,font=\scriptsize,right=2pt]{$x$};

  \end{tikzpicture}
  \hfill~\qed
\end{proof}

  \subsection{Two-path-free graphs}

  \begin{lemma}
    \label{lem:tpf-acyclic}
  Two-path-free graphs are acyclic.
\end{lemma}
\begin{proof}
  If the source and target vertices in the definition of two-path-free graphs are allowed to be the same, this is immediate, as the vertex having a temporal cycle also has an empty walk to itself. If both must be different, we can prove this as follows. By contradiction, suppose a two-path-free graph has a temporal cycle $v_0,v_1,\dots,v_k=v_0$, with $k\geq 2$. Then, $v_1$ can reach $v_0$ through the edge $v_0v_1$ and through the walk $v_1,\dots,v_k$.\qed
\end{proof}

Due to~\Cref{lem:tpf-acyclic}, two-path-free graphs inherit all the properties of acyclic graphs. In addition, they have the following properties.

\begin{lemma}
  \label{lem:tpf-acyclic}
  Two-path-free graphs are minimal.
\end{lemma}
\begin{proof}
  Suppose that a graph $\G$ is two-path-free but not minimal. Thus, there exists an edge $e=uv$ such that $\mathcal{R}(\G\setminus e) = \mathcal{R}(\G)$. The edge $uv$ is a temporal walk from $u$ to $v$ in $\G$, and the previous condition implies that $\G\setminus e$ has another temporal walk from $u$ to $v$. Thus, there are at least two temporal walks from $u$ to $v$ in $\G$, a contradiction.\qed
\end{proof}

\begin{lemma}
  \label{lem:upper-bound-tpf}
  Two-path-free graphs have at most $O(n \log n)$ edges.
\end{lemma}
\begin{proof}
  Let us call the vertices that can reach some vertex $v$ its \emph{predecessors} (initially, itself). The two-path-free property implies that for each edge $uv$, the predecessors of $u$ and $v$ before $\lambda(uv)$ are disjoint (otherwise, the vertices in the intersection can reach $u$ using $uv$ or not using $uv$, using two different walks). Thus, whenever an edge $uv$ occurs, one of $u$ or $v$ at least doubles its number of predecessors. If the number of edges is $\omega(n \log n)$, this doubling event will occur $\omega(\log n)$ times for at least one vertex, contradicting disjointness.\qed
\end{proof}

Observe that this bound is asymptotically tight, as certain happy labelings of hypercubes are two-path-free with $\Theta(n\log n)$ edges.
  
\subsection{Two-path-free \TC graphs}

Due to~\Cref{lem:tpf-acyclic}, two-path-free \TC graphs inherit all the properties of acyclic \TC graphs (as well as the general properties of two-path-free graphs, such as being minimal). In addition, Akos Seress~\cite{seress} characterizes a tight lower bound on the number of edges that a two-path-free graph can have when it is \TC. This is the main result of~\cite{seress}, whose proof is far from immediate, we do not reproduce it.
  
\begin{lemma}[\cite{seress}]
  Two-path-free \TC graphs have at least $2.25n -6$ edges.
\end{lemma}

Regarding lifetime, Douglas West presents in~\cite{west_disjoint} an infinite family of two-path-free \TC graphs whose lifetime $L$ is $n/4+1$. Based on experiments, we conjecture that this lifetime is the maximum possible lifetime of a two-path-free \TC graph.

\begin{conjecture}
  Two-path-free \TC graphs have lifetime $L(\G)\le n/4+1$.
\end{conjecture}

If true, this would imply a nice sequence of lifetime bounds, of at most $n-1$ for acyclic graphs, $n/2$ for acyclic \TC graphs, and $n/4+1$ for two-path-free \TC graphs.

\subsection{\TC graphs}
\label{sec:decomposition}

Let us start with a folklore result from gossip theory.

\begin{lemma}
  If $\G$ is \TC, then $L(\G) \ge \log_2 n$
\end{lemma}
\begin{proof}
  If we see the graph as a sequence of snapshot $G_1, G_2, \dots, G_L$, then the fact that the graph is proper implies that the edges of each snapshot form a matching. Thus, at each time~$t$, the set of nodes that were reached from any fixed vertex can at most double, which implies $L \ge \lceil \log_2 n\rceil$ for \TC. (Note that certain labelings of hypercubes achieve \TC and two-path-freeness with exactly $L = \log_2 n$, so this bound is tight.)\qed
\end{proof}


The remaining content of this section is mostly adapted from the work of Roger Labahn~\cite{labahn93,labahn95}. As a warm-up, note that many papers in the gossip literature analyse the properties of temporal reachability from the point of view of partial orders (posets). It is not mandatory for the reader to be perfectly confident with the notions defined in the next paragraph. We present these notions more as a guide for readers who may want to explore further the gossip literature. The subsequent lemmas are formulated in ways that enable their use without requiring the definition of posets.

Given a graph $\G$, the reachability between the edges of $\G$ naturally defines a partial order such that $e_1 \preceq e_2$ if and only if $e_1 \leadsto e_2$. This is illustrated in a few steps in~\Cref{fig:poset}. From a given graph $\G$, one can look at the local ordering of the edges (or time edges, if the graph is not simple). The line graph of $\G$ can then be oriented according to these local orderings. The resulting graph is a directed acyclic graph (DAG), which contains the same information as this poset. It is often convenient to simplify the poset by focusing on its transitive reduction $\lessdot$, called the \emph{covering relation} of $\preceq$, where $e_1 \lessdot e_2$ if and only if $e_1\ne e_2, e_1 \preceq e_2$ and there is no other $e$ such that $e_1 \preceq e \preceq e_2$. In our context, $e_1 \lessdot e_2$ corresponds to the fact that $e_1$ and $e_2$ are incident to a common vertex $v$ and no other edges $e$ incident to $v$ have an intermediate label between $\lambda(e_1)$ and $\lambda(e_2)$ (these edges are called immediate neighbors in the poset). In particular, the in-degree and out-degree of an element in the covering relation is at most $2$ (at most one for each endpoint of the corresponding edge).

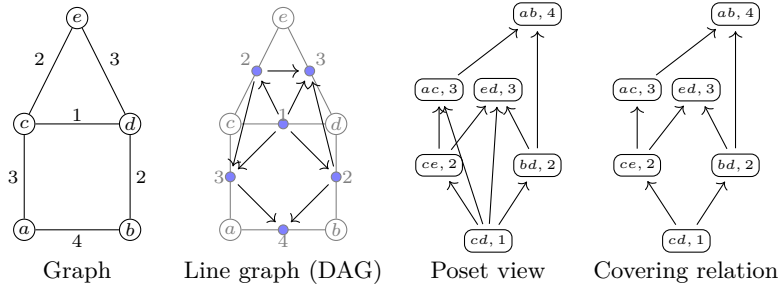
\begin{figure}[h]
  \centering
    \footnotesize
  \begin{tabular}{cccc}
  \begin{tikzpicture}[scale=1.4]
    \tikzstyle{every node}=[defnode,black,minimum size=8pt,font=\scriptsize]
    \path (0,0) node (a) {$a$};
    \path (1,0) node (b) {$b$};
    \path (0,1) node (c) {$c$};
    \path (1,1) node (d) {$d$};
    \path (.5,2) node (e) {$e$};
    \tikzstyle{every node}=[inner sep=2pt,font=\scriptsize]
    \draw (a) -- coordinate (ab) node[below] {4}(b);
    \draw (a) -- node[left] {3}(c);
    \draw (b) -- node[right] {2}(d);
    \draw (c) -- node[above=-1pt] {1}(d);
    \draw (c) -- node[above left] {2}(e);
    \draw (d) -- node[above right] {3}(e);    
  \end{tikzpicture}
    &
  \begin{tikzpicture}[scale=1.4]
    \tikzstyle{every node}=[defnode,minimum size=8pt,font=\scriptsize,fill=white]
    \path (0,0) node (a) {$a$};
    \path (1,0) node (b) {$b$};
    \path (0,1) node (c) {$c$};
    \path (1,1) node (d) {$d$};
    \path (.5,2) node (e) {$e$};
    \tikzstyle{every node}=[inner sep=2pt,font=\scriptsize]
    \draw[gray] (a) -- node[defnode,fill=blue!50,inner sep=1.3pt] (ab) {} node[below] {4}(b);
    \draw[gray] (a) -- node[defnode,fill=blue!50,inner sep=1.3pt] (ac) {} node[left] {3}(c);
    \draw[gray] (b) -- node[defnode,fill=blue!50,inner sep=1.3pt] (bd) {} node[right] {2}(d);
    \draw[gray] (c) -- node[defnode,fill=blue!50,inner sep=1.3pt] (cd) {} node[above=-1pt] {1}(d);
    \draw[gray] (c) -- node[defnode,fill=blue!50,inner sep=1.3pt] (ce) {} node[above left] {2}(e);
    \draw[gray] (d) -- node[defnode,fill=blue!50,inner sep=1.3pt] (de) {} node[above right] {3}(e);

    \tikzstyle{every path}=[shorten >= 2pt,shorten <= 2pt]
    \draw[->] (ac) -- (ab);
    \draw[->] (bd) -- (ab);
    \draw[->] (cd) -- (ac);
    \draw[->] (ce) -- (ac);
    \draw[->] (bd) -- (de);
    \draw[->] (cd) -- (de);
    \draw[->] (ce) -- (de);
    \draw[->] (cd) -- (bd);
    \draw[->] (cd) -- (ce);    
  \end{tikzpicture}
      &
  \begin{tikzpicture}[xscale=1.3]
    \tikzstyle{every node}=[draw,rectangle,rounded corners=3pt,inner sep=2pt,font=\tiny]
    \path (0.5,1) node (1){$cd,1$};
    \path (0,2) node (2a){$ce,2$};
    \path (1,2) node (2b){$bd,2$};
    \path (0,3) node (3a){$ac,3$};
    \path (.6,3) node (3b){$ed,3$};
    \path (1,4) node (4){$ab,4$};

    \tikzstyle{every path}=[shorten >= 2pt,shorten <= 2pt]
    \draw[->] (1) -- (2a);
    \draw[->] (1) -- (2b);
    \draw[->] (1) -- (3a);
    \draw[->] (1) -- (3b);
    \draw[->] (2a) -- (3a);
    \draw[->] (2a) -- (3b);
    \draw[->] (2b) -- (3b);
    \draw[->] (2b) -- (4);
    \draw[->] (3a) -- (4);
  \end{tikzpicture}
        &
  \begin{tikzpicture}[xscale=1.3]
    \tikzstyle{every node}=[draw,rectangle,rounded corners=3pt,inner sep=2pt,font=\tiny]
    \path (0.5,1) node (1){$cd,1$};
    \path (0,2) node (2a){$ce,2$};
    \path (1,2) node (2b){$bd,2$};
    \path (0,3) node (3a){$ac,3$};
    \path (.6,3) node (3b){$ed,3$};
    \path (1,4) node (4){$ab,4$};

    \tikzstyle{every path}=[shorten >= 2pt,shorten <= 2pt]
    \draw[->] (1) -- (2a);
    \draw[->] (1) -- (2b);
    \draw[->] (2a) -- (3a);
    \draw[->] (2a) -- (3b);
    \draw[->] (2b) -- (3b);
    \draw[->] (2b) -- (4);
    \draw[->] (3a) -- (4);
  \end{tikzpicture}\\
    Graph & Line graph (DAG) & Poset view & Covering relation
  \end{tabular}
\caption{\label{fig:poset} Partial order induced by the reachability among edges.}
\end{figure}

In the following, we try to reformulate as much as possible the relevant properties in terms of the original temporal graph, refering to the poset view informally. The important aspect is that the analysis considers mostly reachability between edges rather than vertices. We invite the reader to check~\Cref{lem:basic-reachability} whenever a doubt arises about how reachability among edges transposes to reachability among vertices.

\begin{definition}[Divergent edge] A divergent edge $e$ is a maximal source edge, in the sense that there exists no other source edge $e'$ such that $e\leadsto e'$.
\end{definition}

Note that the set of divergent edges is nonempty, as the source edges contain at least $E^-$ and divergent edges are maximal source edges.
Informally, divergent edges are the last edges that can still reach all the vertices, or equivalently, reach all the edges in $E^+$. Divergent edges are incomparable with each other, by definition. This implies that they also form a matching in the footprint of $\G$. (The situation would be more complex in non-proper graphs, especially in the non-strict setting. How to define analogous concepts in general temporal graphs is an interesting question, not explored here.)

It may appear tempting to define the notion of convergent edges in a symmetrical way (i.e., as minimal sink edges). This was done for example by Krumme~\cite{krumme}. However, we would then miss a nice property of the subsequent decomposition, namely, the fact that all divergent edges can reach all convergent edges (we found some graphs that witness this fact). To prevent this issue, Labahn~\cite{labahn93} defines a more restricted (and asymmetric) concept of convergent edges which we reformulate in two steps as follows.

\begin{definition}[Div-sink edge] A div-sink edge is a sink edge that can be reached by all divergent edges.
\end{definition}

Note that the set of div-sink edges is nonempty, as it contains at least $E^+$. We can now define convergent edges as follows.

\begin{definition}[Convergent edge] A convergent edge $e$ is a minimal div-sink edge, in the sense that there exist no other div-sink edge $e'$ such that $e' \leadsto e$.
\end{definition}

These definitions are summarized in~\Cref{fig:decomposition}.
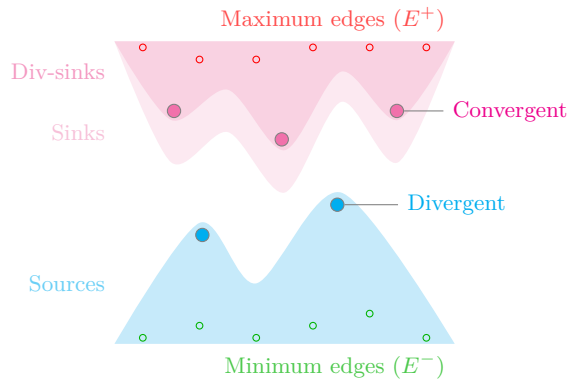
\begin{figure}[h]
  \centering
  \begin{tikzpicture}[xscale=.75,yscale=.8]
    \fill[cyan!20] plot[smooth] coordinates {
      (0,0)
      (1.5,2)
      (2.5,1)
      (4,2.5)
      (6,0)
    }-- (0,0);
    \fill[magenta!10] plot[smooth] coordinates {
      (0,5)
      (1,3)
      (2,3.5)
      (3,2.5)
      (4,4)
      (5,3)
      (6,5)
    };
    \fill[magenta!22] plot[smooth] coordinates {
      (0,5)
      (1,3.7)
      (2,4.2)
      (3,3.2)
      (4,4.5)
      (5,3.7)
      (6,5)
    };

    \tikzstyle{every node}=[defnode,green,minimum size=2.5pt]
    \path (.5,0.1) node (m1){};
    \path (1.5,0.3) node (m2){};
    \path (2.5,0.1) node (m3){};
    \path (3.5,0.3) node (m4){};
    \path (4.5,0.5) node (m5){};
    \path (5.5,0.1) node (m6){};

    \tikzstyle{every node}=[defnode,red,minimum size=2.5pt]
    \path (.5,4.9) node (M1){};
    \path (1.5,4.7) node (M2){};
    \path (2.5,4.7) node (M3){};
    \path (3.5,4.9) node (M4){};
    \path (4.5,4.9) node (M5){};
    \path (5.5,4.9) node (M6){};

    \tikzstyle{every node}=[defnode,fill=cyan,minimum size=5pt]
    \path (1.55,1.8) node (d1){};
    \path (3.93,2.3) node (d2){};

    \tikzstyle{every node}=[defnode,fill=magenta!70,minimum size=5pt]
    \path (1.05,3.85) node (c1){};
    \path (2.95,3.38) node (c2){};
    \path (4.98,3.85) node (c3){};
   
    \tikzstyle{every node}=[font=\footnotesize]
    \path (0,1) node[left,cyan!50]{Sources};
    \path (6,0) node[below left,green!70]{Minimum edges ($E^-$)};
    \path (6,5) node[above left,red!70]{Maximum edges ($E^+$)};
    \path (0,3.5) node[left,magenta!30]{Sinks};
    \path (0,4.5) node[left,magenta!50]{Div-sinks};
    \draw[thin,gray] (c3) -- (5.8,3.85) node[right,magenta]{Convergent};
    \draw[thin,gray] (d2) -- (5,2.3) node[right,cyan]{Divergent};
  \end{tikzpicture}
  \caption{\label{fig:decomposition} The structure of temporal connectivity (poset view).}
\end{figure}
Once again, observe that convergent edges are incomparable by definition, thus they also form a matching. Let us commit these observations for later use into the following lemma.

\begin{lemma}
  \label{lem:div-conv-matching} Divergent edges are incomparable and form a matching in~$\G$. The same holds for convergent edges.
\end{lemma}

Using the above definitions, \TC graphs admit the following characterization:
\begin{lemma} [summarizing from~\cite{labahn93}]
  \label{lem:phases}A graph is \TC if and only if:~\\
  (i) Each vertex can reach at least one divergent edge.\\
  (ii) All divergent edges can reach all convergent edges.\\
  (iii) Each vertex can be reached by at least one convergent edge.
\end{lemma}

Labahn~\cite{labahn93} then proceeds by defining the concepts of kernel and inner kernel as follows.

\begin{definition}[Kernel] The kernel of $\G$ consists of all the edges that both (1) can be reached by at least one divergent edge, and (2) can reach at least one convergent edge. It includes, in particular, the divergent edges and the convergent edges themselves.
\end{definition}

\begin{definition}[Inner kernel] The inner kernel of $\G$ is the part of the kernel that excludes divergent and convergent edges.
\end{definition}

Observe that the three sets corresponding respectively to source edges, div-sink edges, and inner kernel edges in a \TC graph must be disjoint. Yet, they do not form a full partition of $E$, as there may exist edges that are in neither of these sets.

\subsection{Minimal \TC graphs}
\label{sec:minimal-TC}
\noindent

If a \TC graph is minimal, then it must satisfy additional properties.
Lemma~\ref{lem:phases} implies that source edges, div-sink edges, and inner kernel edges are the only ones required for temporal connectivity. All the other edges can be removed without breaking \TC. This yield a nice decomposition (indeed, a partitioning) of the edges of minimal \TC graphs as follows.

\begin{lemma}[adapting from~\cite{labahn93}]
  \label{lem:decomposition}
  In minimal \TC graphs, the edges can be partitioned as:\\
  (i) Source edges (including divergent edges)\hfill (lower part)\\
  (ii) Inner kernel\hfill (middle part)\\
  (iii) Div-sink edges (including convergent edges)\hfill (upper part)
\end{lemma}

Through the next few lemmas, we show that the edges of the lower and upper parts must satisfy even more properties.

\begin{lemma}
  \label{lem:single-walk}
  In a minimal \TC graph, each vertex can reach exactly one divergent edge, and each vertex can be reached by exactly one convergent edge.
\end{lemma}
\begin{proof}
  We prove the lemma for divergent edges. (The arguments for convergent edges are symmetrical.)
  We actually prove something stronger, namely, that each vertex has a single temporal walk that leads to a divergent edge. If not, there must exist somewhere a vertex $u$ that has two local edges, each starting a different temporal walk to some divergent edges (whether both walks end up in the same divergent edge or different ones does not matter). Let $uv$ be the earliest of these two edges. We will show that $uv$ can be removed without breaking \TC.
  First, observe that $uv$ cannot be divergent, as it reaches another divergent edge through $u$. Also observe that $v$ must be able to reach a divergent edge with a departure after $\lambda(uv)$. Thus, both $u$ and $v$ have a temporal walk after $\lambda(uv)$ towards a divergent edge. This implies that any temporal walk from other vertices that would need $uv$ to reach a divergent edge (thus, arriving at $u$ or $v$ before $\lambda(uv)$) can be reconfigured to use either walks instead of crossing $uv$, which preserves property (i) of Lemma~\ref{lem:phases}. The other two properties (ii) and (iii) are unaffected, as they involve only edges that do not reach $uv$ in $\G$. Thus $\G$ remains \TC without $uv$.\qed
\end{proof}

\begin{lemma}
  \label{lem:forest}
  In a minimal \TC graph, the lower part consists of a forest of $k$ disjoint trees of $n-k$ edges in total, where $k$ is the number of divergent edges. The same holds for the upper part with respect to convergent edges.
\end{lemma}
\begin{proof}
  Once again, we prove the lemma for divergent edges. (The arguments for convergent edges are symmetrical.) By~\Cref{lem:single-walk}, the set of vertices $V$ can be partitioned into $V_1,V_2,\dots,V_k$ such that all vertices in $V_i$ can reach only one divergent edge $e_i$. The set of edges of the lower part is similarly partitioned into $E_1, E_2, \dots, E_k$ (each including the corresponding $e_i$). Indeed, if two such sets intersect on some edge $e$, then the endpoints of $e$ can reach several divergent edges (a contradiction). Consider the graph $(V_i,E_i,\lambda_{|_{E_i}})$. For the sake of contradiction, suppose that the footprint $(V_i,E_i)$ has a cycle, and let $uv$ be the edge with smallest label in $E_i$ (breaking ties arbitrarily). All the edges of $E_i$ can reach $e_i$. In particular, the adjacent edges of $uv$ in the cycle can reach $e_i$. But then $u$ (or $v$) has two temporal walks towards $e_i$, contradicting (the proof of)~\Cref{lem:single-walk}. It follows that each graph $(V_i,E_i)$ is a tree of $|V_i|-1$ edges, which concludes the proof.\qed
\end{proof}

\begin{lemma}
  \label{lem:bottom-forest}
  Apart from the divergent edges themselves, every edge of the lower part is needed only in one direction. The same holds for the upper part excluding convergent edges.
\end{lemma}
\begin{proof}
  The lower part consists of disjoint trees for each divergent edge. Every vertex in the trees only needs to reach the corresponding divergent edge (reachability to other vertices is handled in later phases \emph{from} the divergent edge), thus the edges of each tree only need to be crossed \emph{towards} the corresponding divergent edge.\qed
\end{proof}

The situation is summarized in~\Cref{fig:minimal-TC}. As in the previous figure, the picture is drawn from the point of view of the poset. It may therefore appear to be a slight abuse of graphical notation to represent the forest in the lower part of $\G$ as forests in the poset. However, the lower part of the covering relation of the poset is itself a forest, although we do not prove this here, which makes this abuse of notation rather harmless.

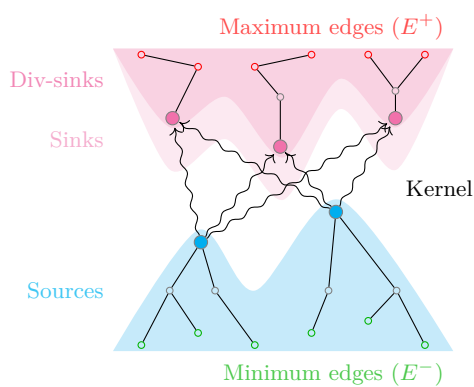
\begin{figure}[h]
  \centering
  \begin{tikzpicture}[xscale=.75,yscale=.8]
    \fill[cyan!20] plot[smooth] coordinates {
      (0,0)
      (1.5,2)
      (2.5,1)
      (4,2.5)
      (6,0)
    }-- (0,0);
    \fill[magenta!10] plot[smooth] coordinates {
      (0,5)
      (1,3)
      (2,3.5)
      (3,2.5)
      (4,4)
      (5,3)
      (6,5)
    };
    \fill[magenta!22] plot[smooth] coordinates {
      (0,5)
      (1,3.7)
      (2,4.2)
      (3,3.2)
      (4,4.5)
      (5,3.7)
      (6,5)
    };

    \tikzstyle{every node}=[defnode,green,minimum size=2.5pt]
    \path (.5,0.1) node (m1){};
    \path (1.5,0.3) node (m2){};
    \path (2.5,0.1) node (m3){};
    \path (3.5,0.3) node (m4){};
    \path (4.5,0.5) node (m5){};
    \path (5.5,0.1) node (m6){};

    \tikzstyle{every node}=[defnode,red,minimum size=2.5pt]
    \path (.5,4.9) node (M1){};
    \path (1.5,4.7) node (M2){};
    \path (2.5,4.7) node (M3){};
    \path (3.5,4.9) node (M4){};
    \path (4.5,4.9) node (M5){};
    \path (5.5,4.9) node (M6){};

    \tikzstyle{every node}=[defnode,minimum size=2.5pt]
    \path (1,1) node (i1){};
    \path (1.8,1) node (i2){};
    \path (3.8,1) node (i3){};
    \path (5,1) node (i4){};
    \path (2.95,4.2) node (I2){};
    \path (4.98,4.3) node (I3){};

    \tikzstyle{every node}=[defnode,fill=cyan,minimum size=5pt]
    \path (1.55,1.8) node (d1){};
    \path (3.93,2.3) node (d2){};

    \tikzstyle{every node}=[defnode,fill=magenta!70,minimum size=5pt]
    \path (1.05,3.85) node (c1){};
    \path (2.95,3.38) node (c2){};
    \path (4.98,3.85) node (c3){};

    \draw (m1) -- (i1);
    \draw (m2) -- (i1);
    \draw (m3) -- (i2);
    \draw (m4) -- (i3);
    \draw (m5) -- (i4);
    \draw (m6) -- (i4);
    \draw (i1) -- (d1);
    \draw (i2) -- (d1);
    \draw (i3) -- (d2);
    \draw (i4) -- (d2);
    \draw (c2) -- (I2);
    \draw (c3) -- (I3);
    \draw (M1) -- (M2);
    \draw (M2) -- (c1);
    \draw (I3) -- (M5);
    \draw (I3) -- (M6);
    \draw (M4) -- (M3);
    \draw (M3) -- (I2);
    
    \draw (d1) edge[decorate, decoration={snake, amplitude=.3mm, segment length=3mm},shorten >=1pt,->] (c1);
    \draw (d2) edge[decorate, decoration={snake, amplitude=.3mm, segment length=3mm},shorten >=1pt,->] (c1);
    \draw (d1) edge[decorate, decoration={snake, amplitude=.3mm, segment length=3mm},shorten >=1pt,->] (c2);
    \draw (d2) edge[decorate, decoration={snake, amplitude=.3mm, segment length=3mm},shorten >=1pt,->] (c2);
    \draw (d1) edge[decorate, decoration={snake, amplitude=.3mm, segment length=3mm},shorten >=1pt,->] (c3);
    \draw (d2) edge[decorate, decoration={snake, amplitude=.3mm, segment length=3mm},shorten >=1pt,->] (c3);

    \tikzstyle{every node}=[font=\footnotesize]
    \path (0,1) node[left,cyan!60]{Sources};
    \path (6,0) node[below left,green!70]{Minimum edges ($E^-$)};
    \path (6,5) node[above left,red!70]{Maximum edges ($E^+$)};
    \path (0,3.5) node[left,magenta!40]{Sinks};
    \path (0,4.5) node[left,magenta!60]{Div-sinks};
    \path (5,2.7) node[right]{Kernel};
  \end{tikzpicture}
  \caption{\label{fig:minimal-TC} General structure of minimal \TC graphs.}
\end{figure}

\begin{conjecture}
  \TC graphs that are minimal and acyclic have at most $O(n\log n)$ edges.
\end{conjecture}

Note that two-path-free graphs are both minimal and acyclic, and they have at most $O(n \log n)$ edges. However, there exist minimal and acyclic \TC graphs which are not two-path-free, so the question is open. Also note that neither acyclicity nor minimality alone guarantee $o(n^2)$ edges, as there exist (non-acyclic) minimal \TC graphs with $\Theta(n^2)$ edges and (non-minimal) acyclic \TC graphs with $\Theta(n^2)$.

\subsection{Additional comments on pivotability and acyclic graphs}

We already saw that acyclic \TC graphs are extremally matched. Using the content of~\Cref{sec:decomposition}, we can actually say something stronger:

\begin{lemma}
  If $\G$ is an acyclic \TC graph, then the divergent edges are exactly $E^-$, and the convergent edges are exactly $E^+$. In other words, the kernel is the entire graph.
\end{lemma}
\begin{proof}
  By contradiction, let $e\notin E^-$ be a divergent edge. The caterpillar from $e$ contains an edge $uv$ of label $1$. This edge is necessarily different from $e$ (as otherwise $e\in E^-$). Thus, $u \leadsto e \leadsto u$ implies a temporal cycle at~$w$.\qed
\end{proof}

\begin{lemma}
  Acyclic \TC graphs (with $n>2$) are not pivotable.
\end{lemma}
\begin{proof}
  If $e=uv$ is a pivot edge, then for all $w\notin \{u,v\}$, we have $w\leadsto e \leadsto w$. As $w$ is not an endpoint of $e$, this implies a temporal cycle for $w$.\qed
\end{proof}

  \begin{lemma}
    \label{lem:minmax}
    Let $\G \in \TC$. If $e= \min(E^+)$ is unique, then $e$ is a pivot edge.
  \end{lemma}
\begin{proof}
  Each vertex can reach $e$ because it is a maximum edge and this edge can reach all vertices because it is unique and all the other maximum edges have a larger label (\Cref{lem:unique}).\qed
\end{proof}

  \begin{lemma}
  Full-range \TC graphs have a pivot edge.
\end{lemma}
\begin{proof}
  If all labels are unique, then $e=\min(E^+)$ is unique. We can thus apply Lemma~\ref{lem:minmax}.\qed
\end{proof}

\begin{lemma}
  If a \TC graph has a single divergent edge, then this edge is a pivot.\\(the same holds if it has a single convergent edge)
\end{lemma}
\begin{proof}
  Let $\G$ have a single divergent edge $e$. As every vertex can reach at least one divergent edge, it can reach $e$, and by definition $e$ can reach all the vertices. Thus, $e$ is a pivot.\qed
\end{proof}

\begin{lemma}
  \label{lem:two-pivots}
  If a \TC graph has more than one pivot edge, then it is not minimal.
\end{lemma}
\begin{proof}
  Suppose the graph has at least two pivot edges $e_1$ and $e_2$. If these edges are comparable, say $e_1 \leadsto e_2$, then $e_2$ can reach at least $n-2$ other edges (the minimum number of edges for $e_1$ to reach all vertices is $n-1$, including $e_2$), and $e_1$ can be reached by at least $n-2$ edges. Overall, this implies that $\G$ has at least $2(n-2)+2=2n-2$ edges. But since the graph is pivotable, there exists a \TC spanner with $2n-3$ edges, thus $\G$ is not minimal. If $e_1$ and $e_2$ are incomparable, then removing either one, say $e_1$, does not impact the reachability $V \leadsto e_2 \leadsto V$, thus $\G$ is again not minimal.\qed
\end{proof}

\begin{lemma}
  Let $\G$ be a minimal \TC graph. Then the number of divergent edges is at most the number of edges with label $1$.
\end{lemma}
\begin{proof}
  Every divergent edge can be reached by at least one edge with label $1$ (caterpillar argument). Thus, if there are more divergent edges than edges labeled $1$, then at least one edge $uv$ labeled $1$ can reach several divergent edges, and so can $u$ and $v$, contradicting Lemma~\ref{lem:single-walk}.\qed
\end{proof}

\begin{lemma}
  Let $\G$ be a \TC graph. If a convergent edge $e_2$ can reach a divergent edge $e_1$, then both $e_1$ and $e_2$ are pivot edges in $\G$. Furthermore, if $\G$ is minimal, this implies $e_1=e_2$.
\end{lemma}
\begin{proof}
  Convergent edges can be reached by all the vertices and divergent edges can reach all vertices. Thus, if $e_2 \leadsto e_1$, then both $e_1$ and $e_2$ are pivot edges. If $e_1\ne e_2$, then $\G$ has two pivot edges, which implies by~\Cref{lem:two-pivots} that $\G$ is not minimal.\qed
\end{proof}


\bibliographystyle{plain}
\bibliography{lemmas}

\appendix
\newpage

\section{Time compression algorithm}
\label{sec:compression}

Given a graph $\G$, the procedure given by Algorithm~\ref{algo:compression} computes a compressed graph $\G'$ whose temporal walks are in bijection with the ones of $\G$. This algorithm is given for happy graphs, but it could easily be extended for more general temporal graphs.

\begin{algorithm}[ht]
  \begin{algorithmic}[1]
  \caption{\label{algo:compression}Time compression}
  \STATE {{\bf Input:} A graph $\G=(V,E,\lambda)$}
  \STATE {{\bf Output:} A time-compressed version of $\G$}
  \STATE {$E' \gets \emptyset$}
  \STATE {$t \gets 1$}
  \WHILE {$E \ne \emptyset$}
  \STATE {$M \gets \{uv \in E \mid uv = e^-(u) = e^-(v)$\}}
  \STATE {$\lambda'(e) \gets t$ for all $e\in M$}
  \STATE {$E \gets E \setminus M$}
  \STATE {$t \gets t + 1$}
  \ENDWHILE
  \STATE {{\bf return} $(V,E',\lambda')$}
\end{algorithmic}
\end{algorithm}

Informally, the algorithm finds all the edges such that no adjacent edges (on either endpoint) are smaller, assigns them label $1$, then removes these edges from the input and repeats with label $2$, and so on. It does so until all the edges have been processed.

\begin{lemma}
  The algorithm preserves happiness (i), preserves the set of temporal walks (ii), produces a compressed graph (iii), and runs in polynomial time (iv).
\end{lemma}
\begin{proof}
  (i) The edges processed in any round form a matching and get a label that is specific to this round (properness). No edge is processed twice (simpleness). \\
  (ii) For any two adjacent edges, the one with smallest time label will be processed at an earlier round than the other. As a result, the procedure preserves local orders among adjacent edges, and preserves the set of temporal walks (up to time distortion).\\
  (iii) By contradiction, let $e$ be an edge such that $\lambda(e) > 1$ and there is no adjacent edge $e'$ such that $\lambda(e')=\lambda(e) -1$. The fact that it does not have an adjacent edge labeled $t-1$ implies that none of its adjacent edges was selected in round $t-1$. As such, $e$ must have been a local minima in round $t-1$ (or before) and have received a label smaller than $t$.\\
  (iv) Every round involves only finding edges that are locally minimum (polynomial time) and the number of round is at most the number of edges (polynomial time).
\end{proof}

\end{document}